\documentclass[11pt]{article}

\def\UseBibLatex{1}

\makeatletter
\def\input@path{{styles/}}
\makeatother

\providecommand{\BibLatexMode}[1]{}
\providecommand{\BibTexMode}[1]{}

\renewcommand{\BibLatexMode}[1]{#1}
\renewcommand{\BibTexMode}[1]{}

\ifx\UseBibLatex\undefined%
  \renewcommand{\BibLatexMode}[1]{}
  \renewcommand{\BibTexMode}[1]{#1}
\fi

\BibLatexMode{%
   \usepackage[bibencoding=utf8,style=alphabetic,backend=biber]{biblatex}%
   \usepackage{sariel_biblatex}%
}

\usepackage[in]{fullpage}%
\usepackage{amsmath}%
\usepackage{amssymb}%
\usepackage[table]{xcolor}%

\usepackage[amsmath,thmmarks]{ntheorem}%

\usepackage{titlesec}%
\usepackage{xcolor}%
\usepackage{mleftright}%
\usepackage{xspace}%
\usepackage{graphicx}
\usepackage{hyperref}%
\usepackage[inline]{enumitem}
\usepackage[ocgcolorlinks]{ocgx2}

\hypersetup{%
      unicode,
      breaklinks=true,%
      colorlinks=true,%
      urlcolor=[rgb]{0.25,0.0,0.0},%
      linkcolor=[rgb]{0.5,0.0,0.0},%
      citecolor=[rgb]{0,0.2,0.445},%
      filecolor=[rgb]{0,0,0.4},
      anchorcolor=[rgb]={0.0,0.1,0.2}%
}

\titlelabel{\thetitle. }%

\theoremseparator{.}%

\theoremstyle{plain}%
\newtheorem{theorem}{Theorem}[section]

\newtheorem{lemma}[theorem]{Lemma}

\newtheorem{corollary}[theorem]{Corollary}

\newtheorem{observation}[theorem]{Observation}

\theoremstyle{plain}%
\theoremheaderfont{\sf} \theorembodyfont{\upshape}%
\newtheorem*{remark:unnumbered}[theorem]{Remark}%

\theoremheaderfont{\em}%
\theorembodyfont{\upshape}%
\theoremstyle{nonumberplain}%
\theoremseparator{}%
\theoremsymbol{\myqedsymbol}%
\newtheorem{proof}{Proof:}%

\providecommand{\emphind}[1]{}%
\renewcommand{\emphind}[1]{\emph{#1}\index{#1}}

\definecolor{blue25emph}{rgb}{0, 0, 11}

\providecommand{\emphic}[2]{}
\renewcommand{\emphic}[2]{\textcolor{blue25emph}{%
      \textbf{\emph{#1}}}\index{#2}}

\providecommand{\emphi}[1]{}%
\renewcommand{\emphi}[1]{\emphic{#1}{#1}}

\definecolor{almostblack}{rgb}{0, 0, 0.3}

\providecommand{\emphw}[1]{}%
\renewcommand{\emphw}[1]{{\textcolor{almostblack}{\emph{#1}}}}%

\providecommand{\emphOnly}[1]{}%
\renewcommand{\emphOnly}[1]{\emph{\textcolor{blue25emph}{\textbf{#1}}}}

\newcommand{\myqedsymbol}{\rule{2mm}{2mm}}
\newcommand{\SarielThanks}[1]{%
   \thanks{%
      Department of Computer Science; %
      University of Illinois; %
      201 N. Goodwin Avenue; %
      Urbana, IL, 61801, USA; %
      \href{mailto:spam@illinois.edu}{sariel@illinois.edu}; %
      \url{http://sarielhp.org/}.%
   #1%
   }%
}

\newcommand{\HLink}[2]{\hyperref[#2]{#1~\ref*{#2}}}
\newcommand{\HLinkSuffix}[3]{\hyperref[#2]{#1\ref*{#2}{#3}}}

\newcommand{\figlab}[1]{\label{fig:#1}}
\newcommand{\figref}[1]{\HLink{Figure}{fig:#1}}

\newcommand{\thmlab}[1]{{\label{theo:#1}}}
\newcommand{\thmref}[1]{\HLink{Theorem}{theo:#1}}

\newcommand{\corlab}[1]{\label{cor:#1}}
\newcommand{\corref}[1]{\HLink{Corollary}{cor:#1}}%

\newcommand{\obslab}[1]{\label{observation:#1}}
\newcommand{\obsref}[1]{\HLink{Observation}{observation:#1}}

\newcommand{\seclab}[1]{\label{sec:#1}}
\newcommand{\secref}[1]{\HLink{Section}{sec:#1}}

\newcommand{\lemlab}[1]{\label{lemma:#1}}
\newcommand{\lemref}[1]{\HLink{Lemma}{lemma:#1}}%

\providecommand{\eqlab}[1]{}%
\renewcommand{\eqlab}[1]{\label{equation:#1}}
\newcommand{\Eqref}[1]{\HLinkSuffix{Eq.~(}{equation:#1}{)}}

\providecommand{\remove}[1]{}%
\newcommand{\Set}[2]{\left\{ #1 \;\middle\vert\; #2 \right\}}

\newcommand{\pth}[1]{\mleft(#1\mright)}%

\newcommand{\ceil}[1]{\mleft\lceil {#1} \mright\rceil}

\newcommand{\brc}[1]{\left\{ {#1} \right\}}

\renewcommand{\Re}{\mathbb{R}}%

\newlist{compactenumA}{enumerate}{5}%
\setlist[compactenumA]{topsep=0pt,itemsep=-1ex,partopsep=1ex,parsep=1ex,%
   label=(\Alph*)}%

\newlist{compactenuma}{enumerate}{5}%
\setlist[compactenuma]{topsep=0pt,itemsep=-1ex,partopsep=1ex,parsep=1ex,%
   label=(\alph*)}%

\newlist{compactenumI}{enumerate}{5}%
\setlist[compactenumI]{topsep=0pt,itemsep=-1ex,partopsep=1ex,parsep=1ex,%
   label=(\Roman*)}%

\newlist{compactenumi}{enumerate}{5}%
\setlist[compactenumi]{topsep=0pt,itemsep=-1ex,partopsep=1ex,parsep=1ex,%
   label=(\roman*)}%

\newlist{compactitem}{itemize}{5}%
\setlist[compactitem]{topsep=0pt,itemsep=-1ex,partopsep=1ex,parsep=1ex,%
   label=\ensuremath{\bullet}}%

\numberwithin{figure}{section}%
\numberwithin{table}{section}%
\numberwithin{equation}{section}%

\newcommand{\eps}{\varepsilon}
\newcommand{\lmax}{l_{max}}
\newcommand{\lcurr}{\Delta_{curr}}
\newcommand{\lfinal}{\Delta_{final}}
\newcommand{\R}{{\cal R}}
\newcommand{\AprxDiameter}{{\tt AprxDiameter}}
\newcommand{\AprxDiamWSPD}{{\tt AprxDiamWSPD}}
\newcommand{\PSET}{P}
\newcommand{\LEFT}{\mathop{\mathrm{left}}}
\newcommand{\RIGHT}{\mathop{\mathrm{right}}}
\newcommand{\parent}{\mathop{\mathrm{parent}}}
\providecommand{\diam}{\Delta}
\renewcommand{\diam}{\Delta}
\newcommand{\Pcurr}{{\cal P}_{curr}}

\newcommand{\U}{\mathcal{U}}
\newcommand{\bd}{{\partial}}%
\newcommand{\CH}{{\mathcal{CH}}}
\newcommand{\etal}{\textit{et~al.}\xspace}

\providecommand{\TPDF}[2]{\texorpdfstring{#1}{#2}}

\BibLatexMode{%
   \bibliography{diameter}
}

\begin{document}
\title{A Practical Approach for Computing the Diameter of a
   Point Set}

\author{%
   Sariel Har-Peled%
   \SarielThanks{}%
}

\date{March 26, 2001\footnote{Re\LaTeX{}ed on \today.}}
\maketitle

\begin{abstract}
    We present an approximation algorithm for computing the diameter
    of a point-set in $\Re^d$.  The new algorithm is sensitive to the
    ``hardness'' of computing the diameter of the given input, and for
    most inputs it is able to compute the {\em exact} diameter
    extremely fast. The new algorithm is simple, robust, has good
    empirical performance, and can be implemented quickly.  As such,
    it seems to be the algorithm of choice in practice for
    computing/approximating the diameter.
\end{abstract}

\section{Introduction}

Given a set of $n$ points $S \subseteq \Re^d$ computing the
diameter of $S$ --- the pair of points of $S$ that are
furthest away from each other --- had attracted a lot of
attention in recent years. While the problem is quite easy
in $2$-dimensions \cite{t-sgprc-83}, even in three
dimensions it becomes non-trivial. Clarkson and Shor
\cite{cs-arscg-89} were the first to give an $O(n \log{n})$
expected time algorithm for computing the diameter in
$3$-dimensions.  However, coming up with a similar
deterministic algorithm proved to be surprisingly hard (see
\cite{r-dadsl-00} for the long and somewhat painful history
of this problem), concluding with the recent result of Ramos
\cite{r-dadsl-00}.

In practice, the diameter is quite useful as it provides a
reliable estimate of the point-set extent, it can be used in
computing a tight fitting bounding box for the point-set
\cite{bh-eamvb-01}, and it can be also used in visualization
and data-mining \cite{fl-fmfai-95}. In particular, in
the graphics community PCA (principal component analysis)
was used to find such tight fitting bounding boxes.
Unfortunately, PCA work reasonably well for the majority of
the inputs but it can be easily fooled.

While theoretically very satisfying, in practice, those
algorithms are quite hard to implement requiring weeks (or
even months) to implement. Furthermore, such implementations
would probably suffer from serious numerical problems. At
the other end of the scale, the trivial algorithm, of trying
all possible pairs and picking the longest pair, is trivial
and can be implemented in a few minutes. However, even for
small inputs it is too slow to be used (for example, for
$10,000$ points in 3d the naive algorithm computes the
diameter in about ten seconds on a Sun Ultra 5).

Motivated by this discrepancy between theory and practice, we present,
in this paper, a simple algorithm for approximating the diameter of a
point-set. For $\eps > 0$, and a set of points $P$ in $\Re^3$, it
computes a pair of points $p,q \in S$, so that
$|p q| \geq (1-\eps)\diam(P)$ in $O((n + 1/\eps^{3})\log{1/\eps})$
time (the algorithm works also in higher dimensions, see
\thmref{bound} for details), where $\diam(P)$ is the length of the
diameter of $P$. This contrast with the faster
$O(n + 1/\eps^{3/2} \log{1/\eps})$ time algorithm due to Barequet and
Har-Peled \cite{bh-eamvb-01} (this algorithm uses the exact algorithm
as a subroutine and is thus impractical), and $O(n + 1/\eps^{5/2})$
algorithm due to Chan \cite{c-adwse-02}\footnote{In fact, Chan's
   algorithm can be speeded up to $O(n + 1/\eps^{2})$ using ideas from
   of \cite{bh-eamvb-01}}.

The new algorithm can also be used to compute the exact
diameter. Although, in the worst case, the algorithm running
time is still quadratic, this running time is sensitive to
the ''hardness'' of the input. Namely, an instance of the
diameter problem is hard, if almost all the pairs of points
in the input have length close to the diameter of the
point-set. However, in most inputs the diameter is quite
distinctive and only a small fraction of all possible pairs
are serious contenders to be the diametrical pair. The
running time of the new algorithm is related to this
quantity and for most inputs it computes the exact diameter
extremely fast.

The idea underlining our algorithm is the following: We compute a
hierarchical decomposition of the given point-set maintaining
implicitly all the pairs of points that might serve as the
diameter. We throw away batches of pairs that are guaranteed not to
contain a long pair.  We repeatedly refine our hierarchical
decomposition (and the associated pairs decomposition) till we reach
the required approximation. In particular, the algorithm works by
computing the well-separated pairs decomposition of the given
point-set \cite{ck-dmpsa-95}, throwing away all the pairs that are
guaranteed to be too short to contain the diameter.  The
implementation of the new algorithm is available on the web
\cite{h-scpca-00}. An implementation in \texttt{Julia} is available
here: \url{https://github.com/sarielhp/WPSD.jl}.

In \secref{prelim}, we define notation and review previous
known approximation algorithms for the diameter.  In
\secref{algorithm}, we present our new algorithm.  In
\secref{analysis}, we provide a theoretical analysis of the
new algorithm.  In \secref{results} we summarize our
experimental results comparing our algorithm with several
other algorithms.  Conclusions are presented in
\secref{conclusions}.

\section{Preliminaries}
\seclab{prelim}

In the following, $P$ is a set of $n$ points in $\Re^d$.
The {\em diameter} of $P$, denoted by $\diam$, is the
maximum distance between a pair of points of $P$.  The two
points $p_\diam, q_\diam \in P$ realizing the diameter of
$P$ are the {\em diametrical pair} of $P$.

\subsection{Fair Split Tree}

In our algorithm, we use the following hierarchical
data-structure known as Fair Split Tree \cite{ck-dmpsa-95}
(a similar notion was previously suggested by Vaidya
\cite{v-oaann-89}).  Generally speaking, Fair-split tree is
an adaptive variant of quadtrees. (Our algorithm can be
described using quadtrees. However, we chose to describe it
using the fair-split tree as in practice it seems to perform
better, resulting in a somewhat more involved analysis.)

Given a point-set $P$ of $n$ points in $\Re^d$, a fair-split
tree is a tree having the points of $P$ stored in its leaves
(a leaf might contain several points of $P$), where each
point of $P$ is stored only once in the tree. For a node
$v$, let $\PSET(v)$ denote the points stored in the subtree of
$v$, and $\R(v)$ denote the minimum axis-parallel bounding
box of $\PSET(v)$.  Let $c(v)$ denote the center of $\R(v)$, and
let $r(v)$ denotes the radius of the minimum ball centered
at $c(v)$ that contains $\R(v)$.

Let $T_0$ be the tree formed by a single node that
corresponds to the whole set $P$.  The tree is constructed
iteratively as follows: In the $i$-th stage one picks a leaf
$u$ of $T_{i-1}$ which contains more than one point of $P$.
Compute the minimum axis parallel bounding box $\R(u)$,
split it in the middle of its longest edge, create two
children $u_l,u_r$ of $u$ that correspond to the two new
boxes, and compute $\PSET(u_l),\PSET(u_r)$ from $\PSET(u)$ and
store them in $u_l,u_r$, respectively.  We will use the
naive implementation of this operation that takes
$O(|\PSET(u)|)$ time (it can be done faster with appropriate
preprocessing; see \cite{ck-dmpsa-95}).

For a node $v \in T$, $\parent(v)$ denotes the parent of $v$
in $T$, and $\lmax(v)$ denotes the length of the longest
edge of $\R(v)$.  A point $p \in P$ {\em belongs} to $v$ if
$p \in \R(v)$.  We associate with a pair of nodes $(u,v)$ of
$T$, the set of pairs of points $\PSET(u,v) = (\PSET(u)
\times \PSET(v)) \cup (\PSET(v) \times \PSET(u))$.  A pair
of points $p,q \in P$ belongs to the pair $(u,v)$ if $(p,q)
\in \PSET(u,v)$.  For a pair $\mu =(u,v)$ of nodes of a
fair-split tree $T$ let
\[
M(\mu) = M(u,v) = |c(u)c(v)| + r_u + r_v
\]
denote the naive upper bound on the maximum distance of any
pair of points of $\PSET(u,v)$.

\remove{ The crucial property of the fair-split tree that we
   need is the following packing property:

   \begin{lemma}[\cite{ck-dmpsa-95}, Lemma 4.1]
       Let $T$ be a fair split tree associated with $P$, $C$
       a cube, $S = \brc{v_1, \ldots, v_l}$ be a set of
       nodes in $T$ such that $\PSET(v_i)\cap \PSET(v_j) =
       \emptyset$ for all $i \ne j$, $\lmax(\parent(v_i)) \geq
       \lmax(C) /c$, $\R(v_i) \cap C \ne \emptyset$. Let
       $K(c,d)$ be the maximum number of elements of $S$. We
       have, $K(c,d) \leq (3 c + 2)^d$.
   \end{lemma}

   Intuitively, \lemref{packing} states that the number of
   nodes of a fair-split tree that are roughly of the same
   size which are close to each other is limited (i.e., it
   is impossible to have a huge number of very thin and long
   cells close to each other).
   }

\subsection{Review of Known Approximation Algorithms for
   computing the Diameter}

In this section, we review the known approximation
algorithms to the diameter. We implemented most of them and
compared them to our new algorithm.  Approximation
algorithms for the diameter were previously suggested by
Agarwal \etal  \cite{ams-fnmst-92}, Barequet and Har-Peled
\cite{bh-eamvb-01} and Chan \cite{c-adwse-02}.  Barequet and
Har-Peled algorithm is a folklore result, although we are
not aware of earlier reference.

\paragraph{Axis parallel bounding box}
Let $\R = \R(P)$ denote the axis parallel bounding box of
$P$. Let $p,q \in P$ be the two points that lie on the two
sides of $\R$ furthest away from each other.  Clearly, $|p q|
\leq \diam \leq \sqrt{d}|p q|$. Since $B$ can be computed in
linear time, it follows that we can compute a constant
approximation to the diameter of $P$ (i.e., $p,q$) in linear
time.

\paragraph{Snapping $P$ to a Grid}
We decompose $\R = \R(P)$ into a grid $G$ of $k \times k
\times \ldots \times k$ cells, where $k = \ceil{1/\eps}$.
Each cell is a shrunk copy of $\R$ by a factor of $k$. Let
$P_G'$ be the set of the centers of the cells of $G$ that
contains points of $P$ in their interior.  Note, that $P_G'$
might contain more than two points on an axis-parallel line
$\ell$. Since the pair of points that realize the diameter
are pair of vertices of $\CH(P_G)$, it is clear that we can
remove all the points of $P_G'$ on such a line $\ell$ which
are not extreme. Repeating this process for all the
axis-parallel lines that pass through points of $P_G'$,
reduces the number of points from $O(1/\eps^d)$ to
$O(1/\eps^{d-1})$, and let $P_G$ denote the resulting set of
points. We call this process {\em grid cleaning}. It takes
$O(n + 1/\eps^{d-1})$ time. Being slightly more careful
about the implementation, the running time can be reduced to
$O(n)$.  Clearly,
\[
(1-\eps)\diam(P) \leq \diam(P_G) \leq (1+\eps)\diam(P).
\]
Computing the diameter of $P_G$ in a brute force manner takes
$O(1/\eps^{2(d-1)})$ time. As observed by Barequet and Har-Peled
\cite{bh-eamvb-01}, by using known bounds on the number of vertices of
convex hulls of grid points, together with output-sensitive
convex-hull algorithms, this result can be slightly improved. See
\cite{bh-eamvb-01} for the details. For $d=3$ case, \cite{bh-eamvb-01}
present an algorithm that runs in $O(n + 1/\eps^{3/2})$, which is only
of theoretical intersect as it uses the algorithm of Clarkson and Shor
\cite{cs-arscg-89}.

\paragraph{Searching for the Diameter Direction}
Let $\ell$ be a line, and let $P_\ell$ be the projection of
$P$ into $\ell$. Clearly, the distances between a pair of
projected points of $P_\ell$ is smaller than the original
distance of the pair. In particular, $\diam(P_\ell) \leq
\diam(P)$. However, if the angle $\alpha$ between $\ell$ and
the diametrical pair $p_\diam q_\diam$ is smaller than
$\sqrt{2\eps}$ then $\diam(P_\ell) \geq |p_{\diam}
q_{\diam}|\cos(\alpha) \geq (1 -\alpha^2/2) \diam \geq
(1-\eps)\diam$. Thus, by covering the sphere of directions
in $\Re^d$ by $O(1/\eps^{(d-1)/2})$ caps of angular radius
$\leq \sqrt{2\eps}$, and picking a vector inside each cap,
we have a set $V$ of $O(1/\eps^{(d-1)/2})$ directions, so
that $p_\diam q_\diam$ is an angular distance $\leq
\sqrt{2\eps}$ from one of those directions. In particular,
projecting the points of $P$ into each of the directions of
$V$, finding the extreme pair in each direction, and
taking the pair of maximum distance, results in an
$\eps$-approximation to $\diam$. The running time is
$O(n/\eps^{(d-1)/2})$.

By first doing grid cleaning (with $\eps/2$) and then using
the projection algorithm, we have an $\eps$-approximation to
the diameter that works in $O(n + 1/\eps^{3(d-1)/2})$ time.

\paragraph{Reducing the number of projections}
Chan \cite{c-adwse-02} observed that the bottleneck in the
above algorithm is the large number of projection it
performs. Instead of projecting the points to lines,
project the points into hyperplanes, and compute the
approximate diameter of the projected point-set. One can
easily find a family of $O(1/\sqrt{\eps})$ hyperplanes, so
that the angular distance between the diameter and one of
the hyperplanes is smaller than $\sqrt{\eps}$. Now compute
recursively, for each projected point-set its $\eps/2$
approximate diameter. The maximum diametrical pair computed
in the recursive calls is an $\eps$-approximation to the
diametrical pair.

Observing that one can perform grid-cleaning before
projecting the points (and thus also in the recursive calls
on the $(d-1)$-dimensional projected points) results in an
algorithm with $O(n +1/\eps^{d-1/2})$ running time.  By
doing more aggressive grid-cleaning (using convex-hull
algorithms in two and three dimensions) one can reduce the
running time to $O(n + 1/\eps^{d-1})$.  See
\cite{c-adwse-02}.

\section{The Algorithm}
\seclab{algorithm}

Let $P$ be a point-set in $\Re^d$, $\eps>0$ a prescribed
approximation parameter, and $T$ be a fair-split tree of
$P$. We will compute parts of $T$ as the algorithm
progresses in an on-demand fashion, starting from a tree
having a single node that contains all the points of $P$.

The algorithm works by maintaining a current approximation
to the diameter, denoted by $\lcurr$, and a set $\Pcurr$ of
pairs of nodes of $T$, so that all the pairs of points that
(substantially) might improve our approximation must belong
to one of those nodes pairs. Such a node pair might couple a
node with itself (i.e., in the beginning we couple the root
of $T$ with itself as all the $\Theta(n^2)$ pairs of points
are candidates to be the diametrical pair).

The algorithm initializes $\lcurr$ to be the side length of
$\R({root(T)})$ (which is a $\sqrt{d}$ approximation to the
diameter). And $\Pcurr = \brc{ (root(T), root(T)) }$.  The
algorithm associates the value $M(u,v)$ with a pair $(u,v)$,
and it maintains a heap of the pairs of $\Pcurr$, sorted by
a decreasing order by this value. In each iteration, the
algorithm picks the maximum element in the heap (i.e., the
element $(u,v)$ in $\Pcurr$ with maximum $M(u,v)$, as
intuitively, this is the pair with the highest potential to
improve our current estimate of the diameter), inspect it,
and either throws it away, or generates $O(1)$ pairs from
it, and stores them in $\Pcurr$ (updating $\lcurr$ in the
process).

In particular, a pair $(u,v)$ can be thrown away if $\PSET(v)$ or
$\PSET(u)$ are empty, or if
\begin{equation}
    M(u,v) \leq (1+\eps)\lcurr, \eqlab{cond}
\end{equation}
as $M(u,v)$ is an upper bound on the distance of any pair of points in
$\PSET(v) \times \PSET(u)$, and the above condition states that all
pairs is $\PSET(u,v)$ are within $1+\eps$ of our current estimate to
the diameter.

\newcommand{\XIpe}[2]{\includegraphics{#1}}%
\begin{figure}
    \begin{tabular}{cccc}
        \XIpe{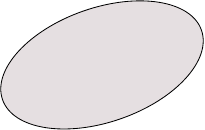}{figs/diam01.ipe}
        &
        \XIpe{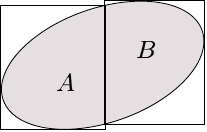}{figs/diam02.ipe}
        &
        \XIpe{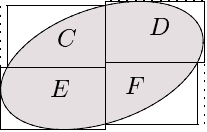}{figs/diam03.ipe}
        &
        \XIpe{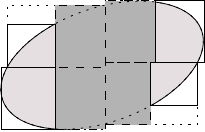}{figs/diam04.ipe}
        \\
        (a) & (b) & (c) & (d)\\
        \\
        \XIpe{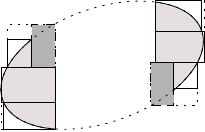}{figs/diam05.ipe}
        &
        \XIpe{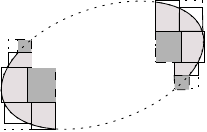}{figs/diam06.ipe}
        &
        \XIpe{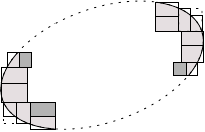}{figs/diam07.ipe}
        \\
        (e) & (f) & (g)\\
    \end{tabular}
    \caption{Illustration of how \AprxDiameter{} works on a
       densely sampled ellipse. As the algorithm progresses
       the pairs-decomposition become finer, and pairs that
       are too short are being thrown away. The dark regions
       are nodes that all the pairs associated with them
       were thrown away. As time progresses, the pairs
       maintained by the algorithm converge to the real
       diameter. It seems that for all ``real'' inputs,
       after several iterations, only a small fraction of
       the input points would be contained in the active nodes
       maintained by the algorithm.}

    \figlab{example}
\end{figure}

If the algorithm decides not to throw away a pair $(u,v)$,
it expands it by all the pairs formed by the children of $u$
and $v$.  Namely, $(u,v)$ generates the four pairs
$(\LEFT(u), \LEFT(v))$, $(\LEFT(u), \RIGHT(v))$,
$(\RIGHT(u), \LEFT(v))$, $(\RIGHT(u),\RIGHT(v))$ (however,
if $u=v$ the pairs
\begin{equation*}
    (\LEFT(u),\RIGHT(v)), (\RIGHT(u),
    \LEFT(v))
\end{equation*}
are identical, and only three pairs should be generated in such a
case), where $\LEFT(u),\RIGHT(u)$ denotes the left and right child of
$u$ respectively. Note, that if $\PSET(u)$ (or $\PSET(v)$) is a
singleton we do not split it, splitting the pair only on one side of
it.  We insert the new pairs into $\Pcurr$.

Whenever the algorithm creates a pair $(u,v)$, it also
picks two points $p \in \PSET(u), q \in \PSET(v)$, computes $|p q|$,
and update $\lcurr$ if $|p q| > \lcurr$. Furthermore, a new
pair is being added to $\Pcurr$ only if it passes the
condition of \Eqref{cond}.  This guarantees that as the pair
decomposition becomes finer also the current approximation
improves.

As we descend in $T$ from $u$ to $\LEFT(u)$, $\RIGHT(u)$ in
the expansion process, it might be that $\LEFT(u)$ or
$\RIGHT(u)$ do not exist yet in $T$, as the node $u$ was not
split yet.  In this point in time, we split $u$ creating
$\LEFT(u)$, $\RIGHT(v)$ (splitting the set $\PSET(u)$ into
$\PSET(\LEFT(u)), \PSET(\RIGHT(u))$). We argue below that
our algorithm constructs only a small fraction of the whole
fair-split tree of $P$.

The algorithm stop when $\Pcurr$ becomes empty and returns
the value of $\lcurr$.  We call this algorithm
$\AprxDiameter{}$.  We modify slightly this algorithm, as
follows: When splitting a pair $(u,v)$ instead of generating
four pairs, we split only the node in the pair having the
larger bounding box (i.e., for a pair $(u,v)$ we split $u$
if $\lmax(u) > \lmax(v)$), generating only two new pairs.
The advantage of the new approach is that, although pairs
might still be unbalanced (i.e., the sizes of their
corresponding bounding boxes are of different orders of
magnitude), the pair is balanced if we refer to the node's
parent instead of the node itself. Let $\AprxDiamWSPD$
denote the new algorithm. This is how the well-separated
pairs decomposition is constructed by the algorithm of
\cite{ck-dmpsa-95}.

\figref{example} illustrates how the algorithm works, and
why in practice the algorithm works well: For most inputs
only a small fraction of the nodes of the fair-split tree
are active (i.e., participate in a pair which is stored in
$\Pcurr$) after several iterations. In particular, only
small fraction of the tree is being constructed by the
algorithm, and similarly, only a small fraction of all
possible pairs are considered by the algorithm.

\begin{lemma}
    The algorithms $\AprxDiameter$, $\AprxDiamWSPD$ always
    stop, and returns a number $D$, such that
    $(1-\eps)\diam \leq D \leq \diam$. Moreover, the running
    time of the algorithms is $O(n^2 \log{n})$ for any value
    of $\eps$.

    \lemlab{diam:alg}
\end{lemma}

\begin{proof}
    Note, that when a node is split in the fair-split tree, at least
    one point of $P$ goes to each child.  Thus, there are $O(n)$ total
    nodes in $T$. It is easy to verify by induction that the algorithm
    never generates the same pair twice. Finally, if the algorithm
    encounters a pair $(u,v)$ such that $\PSET(u),\PSET(v)$ are a
    singletons, then the algorithm computes the length of the
    corresponding pair of points, updates $\lcurr$, and then the pair
    fails the condition of \Eqref{cond} (indeed, $M(u,v)$ is no more
    than the length of the corresponding pair, as $r_u = r_v = 0$),
    and the algorithm throws this pair away.

    As for the running time, we note that the time spent on
    constructing the whole fair-split tree is $O(n^2)$ in
    the worst case. Since splitting and handling each pair
    takes $O(\log{n})$ time (i.e., finding the maximal
    element in the heap of pairs), and there are $O(n^2)$
    pairs, it follows that the running time of the algorithm
    is $O(n^2 \log{n})$, as $O(n^2)$ insertions/deletions
    are being performed on the heap.

    As for the quality of approximation: Let $p_\diam, q_\diam$ be the
    diametrical pair of $P$, and let $(u,v)$ be the pair constructed
    by the algorithm such that $(p_\diam, q_\diam) \in \PSET(u,v)$ and
    $(u,v)$ fails the condition of \Eqref{cond}. Furthermore, let
    $\lcurr'$ denote the value of $\lcurr$ when the condition failed,
    and let $\lfinal$ denote the final value of $\lcurr$ returned by
    the algorithm. Clearly,
    \[
    (1+\eps)\lfinal \geq (1+\eps) \lcurr' \geq M(u,v) \geq
    |p_\diam q_\diam| = \diam.
    \]
    Thus,
    \[
    \lfinal \geq \frac{\diam}{1 + \eps} \geq
    \frac{\diam}{1 + \eps + \eps^2 + \cdots} = (1-\eps)\diam.
    \]
\end{proof}

\section{Theoretical Analysis of \AprxDiameter{} and
   \AprxDiamWSPD{}}
\seclab{analysis}

In this section, we prove that the running times of
$\AprxDiameter$ and $\AprxDiamWSPD$ are
$O((n+1/\eps^{2d})\log{1/\eps})$ and
$O((n+1/\eps^{3(d-1)/2})\log{1/\eps})$, respectively. The
less theoretically inclined reader is encouraged to skip
this section.

\begin{corollary}
    Let $(u_1, v_1), (u_2, v_2), \ldots, (u_m,v_m)$ the
    sequence of pairs handled by  \AprxDiameter{}. Then
    $M(u_1, v_1) \geq M(u_2, v_2) \geq \ldots \geq M(u_m,
    v_m)$.  \corlab{sorted}
\end{corollary}

\begin{proof}
    Observe that if a pair $(w,x)$ was expanded from $(u,v)$
    than $M(u,v)\geq M(w,x)$. Since the algorithm always
    pick the pair in $\Pcurr$ that maximizes the value of
    $M(\cdot, \cdot)$ implies, by induction, the statement.
\end{proof}

We next improve the analysis of \lemref{diam:alg}. For a pair
$\mu = (u,v)$ let
$\lmax(\mu) = \lmax(u,v) = \max \pth{ \lmax( \R(u)),
   \lmax(\R(v))}$. Let $\U$ denote the set of pairs created by
\AprxDiameter{} during its execution. We partition $\U$ into classes:
\[
    S_i = \Set{ (u,v) }{ (u,v) \in \U, \frac{\diam}{2^{i+1}} <
       \lmax(u,v) \leq \frac{\diam}{2^i}},
\]
for $i \geq 1$. Let $r_i$ denote the maximum radius of nodes
of pairs of $S_i$:
\[
r_i = \max_{(u,v) \in S_i} \max( r(u), r(v)).
\]

\begin{observation}
    $\diam/2^{i+1} \leq r_i \leq \sqrt{d} \diam/2^{i}$.
    \obslab{r:i}
\end{observation}

\begin{lemma}
    For $i > K = \ceil{ \log \frac{ 4 d}{\eps}}$, the sets
    $S_i$ are empty. This holds for the pairs generated by
    the algorithms \AprxDiameter{} and \AprxDiamWSPD{}.

    \lemlab{S:i:empty}
\end{lemma}

\begin{proof}
    Let $(u,v)$ be a pair of $S_i$, and $\lcurr'$ be the
    value of $\lcurr$ when $(u,v)$ was inserted into
    $\Pcurr$.  Clearly, $M(u,v) \leq \lcurr' + 4r_i$, since
    $\lcurr'$ was updated by a pair of points $p \in
    \PSET(u), q \in \PSET(v)$ before the pair was created
    (i.e., $\lcurr' \geq |p q|$). On the other hand, the
    pair $(u,v)$ was inserted into the set $\Pcurr$, which
    implies that $M(u,v) > (1+\eps)\lcurr'$. Furthermore,
    $\lcurr' \geq \diam/\sqrt{d}$ (as $\lcurr$ initial value
    was no smaller than this, and its value only increase
    during the algorithm execution).  Thus,
    \[
    r_i \leq \frac{\sqrt{d} \diam}{2^{i}} \leq
    \frac{\eps}{4d} \cdot \sqrt{d}\diam
    = \eps \frac{\diam}{4 \sqrt{d}} \leq \frac{\eps\lcurr'}{4},
    \]
    and
    \[
    \lcurr' + 4r_i \geq M(u,v)  > (1+\eps)\lcurr' \geq
    \lcurr' + 4(\eps\lcurr'/4) \geq \lcurr' + 4r_i.
    \]
    A contradiction.
\end{proof}

We observe, that the children of a pair $(u,v) \in S_i$
might still be in $S_i$. However, after at most $d$ splits
the descendant pair must belong to a class $S_j$, where $j >
i$.  Indeed, each split shrinks by half the longest
dimension of the bounding boxes on both sides of the pair.
After $d$ splits, the maximum side of the a box must have
shrunk by $2$, excluding it from $S_i$. Since the algorithm
compute the minimum axis-parallel bounding box of a node
during its construction, a pair generated from a pair $\mu$
might be several classes away from $\mu$.

The pairs generated from pair $(u,v)$ of depth $l$ in $T$
(i.e., the maximum depth of $u$ and $v$ is $l$) have depth
at most $l+1$. Since there are only $\ceil{ \log \frac{ 4
      d}{\eps}}$ non-empty classes, and at most $d$
consecutive elements on a path of $T$ might belong to the
same class. We conclude that the deepest node visited by the
algorithm is of depth $\leq d\ceil{ \log \frac{ 4
      d}{\eps}}$. We conclude:

\begin{corollary}
    The algorithm $\AprxDiameter$ constructs at most
    $O(1/\eps^d)$ nodes of $T$. In particular, the running
    time of $\AprxDiameter$ is
    $O((n+1/\eps^{2d})\log{1/\eps})$.

    \corlab{run:aprx:diam}
\end{corollary}

\begin{proof}
    The fair-split tree $T$ is a binary tree. By the above
    discussion, the deepest node constructed is of depth
    $\leq d\ceil{ \log \frac{ 4 d}{\eps}}$. Thus, the
    overall number of nodes constructed is
    \[
    O \pth{ 2^{d\ceil{ \log \frac{ 4
      d}{\eps}}} } = O \pth{ \pth{\frac{4d}{\eps}}^{d}} = O
      \pth{ \frac{1}{\eps^d}}.
    \]
    Thus, the overall number of pairs constructed by the
    algorithm is $O(1/\eps^{2d})$. Furthermore, each
    operation on the heap now takes only $O(\log{1/\eps})$
    time.  Constructing each level in $T$ takes $O(n)$
    times, and the algorithm constructs $O(\log{1/\eps})$
    levels. We conclude that the overall running time of the
    algorithm is $O((n+1/\eps^{2d})\log{1/\eps})$.
\end{proof}

Unfortunately, \corref{run:aprx:diam} is almost tight. The
reason being that pairs might have a node with a huge
bounding box on one side, and a tiny bounding box on the
other side. Although this is not a ``natural'' input, one
can easily construct such an input for which the running
time is close to the bound stated in \corref{run:aprx:diam}.
We thus turn our attention to \AprxDiamWSPD{}.

\begin{observation}
    \obslab{big:parents}%
    If a pair $(u,v)$ computed by $\AprxDiamWSPD$ belongs to
    the class $S_i$, then
    \begin{equation*}
        \lmax(\parent(u)),
        \lmax(\parent(v)) > \diam/2^{i+1}.
    \end{equation*}
\end{observation}

\begin{lemma}[\cite{ck-dmpsa-95}, Lemma 4.1]
    Let $T$ be a fair split tree associated with $P$, $C$ a
    cube, $S = \brc{v_1, \ldots, v_l}$ be a set of nodes in
    $T$ such that $\PSET(v_i)\cap \PSET(v_j) = \emptyset$ for all $i
    \ne j$, $\lmax(\parent(v_i)) \geq \lmax(C) /c$, $\R(v_i) \cap
    C \ne \emptyset$. Let $K(c,d)$ be the maximum number of
    elements of $S$. We have, $K(c,d) \leq (3 c + 2)^d$.
    \lemlab{ck:ratio}
\end{lemma}

Intuitively, \lemref{packing} states that the number of
nodes of a fair-split tree that are roughly of the same size
which are close to each other is limited (i.e., it is
impossible to have a huge number of very thin and long cells
close to each other).  In the following, let $N_i$ denote
the set of all the nodes of the fair split tree that
participate in pairs of $S_i$.

We bound the running time of \AprxDiamWSPD{} by bounding the
number of pairs created by the algorithm. Our proof relies
on the following observations: (i) the nodes of $N_i$ are
not densely packed together (\lemref{N:i:packing}), (ii) all
the pairs of $S_i$ are of similar length
(\lemref{similar:len}), (iii) The nodes of $N_i$ lie
``close'' to the boundary of the convex-hull of $P$
(\lemref{close:to:boundary}), (iv) no two pairs of $S_i$ can
be too ``far'' from each other (\lemref{packing}), and (v)
clustering the pairs by directions and proximity, we derive
our bound on the number of pairs generated by the algorithm
(\lemref{overall:pairs}).

\begin{lemma}
    Let $T$ be a fair split tree associated with $P$, $C$ a
    cube, $S = \brc{v_1, \ldots, v_l} \subseteq N_i$ be a
    set of nodes in $T$, $\R(v_j) \cap C \ne \emptyset$, for
    $j=1,\ldots, l$. Let $K(r,d,i)$ be the maximum number of
    elements of $S$, where $r = \lmax(C)$.  We have,
    \[
    K(r,d, i) = O \pth{ \pth{ \frac{r2^i}{\Delta}}^d}.
    \]
    \lemlab{N:i:packing}
\end{lemma}

\begin{proof}
    Let $u,v$ be two nodes of $S$. If $u$ is an ancestor of
    $v$ in $T$, we claim that the distance in $T$ between
    the two nodes is at most $2d + 1$. Indeed, by
    \obsref{big:parents}, $\lmax(\parent(v)) > \Delta/2^{i+1}$. In
    particular, let $v_a$ be the ancestor of $v$ in distance
    $2d+1$ from it. Since each time we climb a level in $T$,
    the size of one of the coordinates double at least, and
    this new value is larger or equal to the previous
    longest dimension, it follows, that $\lmax( \cdot )$ at
    least doubles when we climb $d$ levels. It follows that
    $\lmax(v_a) > \Delta/2^{i-1}$. However, $v_a$ can not
    possibly participate in a pair of $S_i$ (i.e., its
    bounding box $\R(v_a)$ is too large). Implying that the
    distance between $u$ and $v$ in $T$ is at most $2d +1$.

    Let $S'$ be the set of all nodes of $S$, such that none
    of their ancestors is in $S$ (this is the set of
    ``maximal'' nodes in $S$). By the above argument, and
    since $T$ is a binary tree, it follows: $|S| \leq
    |S'|2^{2d+1} = O(|S'|)$.  Furthermore, all the nodes of
    $S'$ corresponds to rectangles which are disjoint and
    not intersecting.  Note that for $v \in S'$, we have
    $\lmax(v) \geq \Delta/2^{i+1}$, and
    \[
    \frac{\lmax(C)}{\lmax(v)} \leq \frac{r2^{i+1}}{\Delta}.
    \]
    In particular, by \lemref{ck:ratio}, we have
    \[
    |S'| \leq \pth{\pth{ \frac{3r 2^{i+1}}{\Delta} + 2 }^d}
    = O \pth{ \pth{ \frac{r2^i}{\Delta}}^d},
    \]
    which implies the result as $|S| = O(|S'|)$.
\end{proof}

\begin{lemma}
    Given a cube $C$ of side length $\ell_0$ partitioned into
    $k^d$ equal size cubes by a grid of cubes $G$, and $D$
    be convex shape. The number of cells of $G$ in distance
    at most $L$ from any point of the boundary of $D$ is
    $O(k^{d-1}(1 + (k L/\ell_0)^d))$.

    \lemlab{grid:surface}
\end{lemma}

\begin{proof}
    This is a standard argument about the relation between a
    surface area of a convex shape, and the number of grid
    cells that it intersects (in fact, stronger bounds can
    be proved).  For $x=0$ this is no more than the number
    of grid cells of $G$ that intersects $\bd{D}$.  Let $C'$
    be a cell of $G$ that intersects $\bd{D}$. Let $f$ be
    the highest dimensional feature of $\bd{C'}$ that
    intersects $\bd{D}$. One can charge this intersection to
    the full dimensional flat that supports $f$. Since the
    number of such flats is $O(k^{d-1})$ as can be easily
    verified, and each such flat is charged at most twice.
    It follows that the number of such grid cells is
    $O(k^{d-1})$.

    As for any value of $x$, we charge a grid cell $C'$
    in distance $\leq x$ from $\bd{D}$ to the grid cell that
    contains this nearest point. Each grid cell that
    intersects $\bd{D}$ is being charged $O((k L/\ell_0)^d)$
    times, and the result follows.
\end{proof}

Let $S_i'$ be the subset of pairs of $S_i$ that were
expanded by $\AprxDiamWSPD$. In the following, we charge the
elements of $S_i$ to their parents (a pair can be charged at
most $O(1)$ times), and we bound the number of elements in
$S_i'$.

Let $M_i$ denote the maximum of $M(u,v)$ overall pairs
$(u,v) \in S_i'$, and let $\mu_i$ the pair that realizes
$M_i$.  Clearly, by \corref{sorted}, $\mu_i$ is the first
pair of $S_i'$ being handled by $\AprxDiamWSPD$. In
particular, when $\mu_i$ is being handled the value of
$\lcurr$, denoted by $\lcurr^i$, is at least $M_i - 4r_i$.
Furthermore, for all $(u,v) \in S_i'$, we have
\begin{equation}
    M(u,v) \geq
    \lcurr \geq \lcurr^i \geq M_i
    -4r_i. \eqlab{bound:on:m:i}
\end{equation}

Note, that if any of the pairs $(u,v)$ of $S_i$ is such that
$\PSET(u,v)$ contains the diametrical pair of $P$, it follows that
$M_i \geq \diam$. Otherwise, the diametrical pair was thrown away
together with a pair $(u,v)$. But arguing as in the proof of
\lemref{diam:alg}, it follows that
$M_i \geq \lcurr^i \geq (1-\eps)\diam$.  Furthermore, since $\mu$ was
expanded by the algorithm, it means that condition \Eqref{cond}
failed, implying that
$M_i \geq (1+\eps)(1-\eps)\diam = (1 - \eps^2)\diam$. Thus,
$(1-\eps^2)\diam \leq M_i \leq \diam +4r_i$. However, by
\lemref{S:i:empty}, if $S_i'$ is not empty then
$16d r_i \geq \eps \Delta$. In particular,
$\diam - 16 d r_i \leq M_i \leq \diam + 4r_i$. We conclude:

\begin{lemma}
    If the set $S_i'$ is not empty, then for any pair $(u,v)
    \in S_i'$ we have
    \[
    \diam -(16 d + 4)r_i \leq M(u,v) \leq \diam + 4r_i,
    \]
    and, in particular, $\diam \leq M_i + (16d + 4)r_i$.

    \lemlab{similar:len}
\end{lemma}

\begin{lemma}
    Let $\CH$ be the convex hull of $P$, and let $v$ be a
    node of $N_i$. The distance of $\R(v)$ from the boundary
    of $\CH$ is at most $21d r_i$.
    \lemlab{close:to:boundary}
\end{lemma}

\begin{proof}
    Indeed, let $(u,v)$ be a pair of $S_i$ that contains
    $v$.  Clearly, if the distance between $\R(v)$ and the
    boundary of $\CH(P)$ is larger than $21d r_i$ then
    \[
    \diam = \diam(\CH(P)) > M(u,v) -2r_i + 21d r_i,
    \]
    as we can pick two points of $\R(u), \R(v)$ that lies
    inside $\CH(P)$ of distance $\geq M(u,v) - 2r_i$, and by
    the above condition, we can extent this segment by
    length $21d r_i$ while remaining inside $\CH(P)$.

    This implies that
    \[
    \diam >
    M(u,v) + 21d r_i -2r_i \geq M_i + 21d r_i -6r_i
    \geq M_i + 18r_i  \geq M_i + 16r_i +4r_i,
    \]
    since $d \geq 2$, and by \Eqref{bound:on:m:i}.  Contradiction to
    \lemref{similar:len}.
\end{proof}

In the following, we will bound the number of pairs
generated by $\AprxDiamWSPD$ by a packing argument.

\begin{figure}
    \centerline{\XIpe{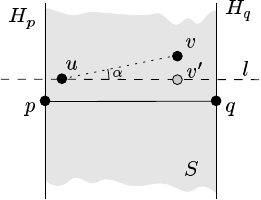}{figs/angle.ipe}}
    \caption{Illustration of the proof of \lemref{packing}}
    \figlab{angle}
\end{figure}

\begin{lemma}
    Let $\delta, 0 < \delta < 0.5$ be a given parameter, and
    $u,v,p,q$ four points taken from a set $Q, Q \subseteq
    \Re^d$. And we assume that the following holds:
    \begin{itemize}
        \item $p q, u v$ are ``almost'' diametrical pair:
        $1 - \delta \leq |p q|, |u v| \leq 1 + \delta$.

        \item None of the four points are too close to each
        other: $|p u|, |p v|, |q u|, |q v| \geq 3 \sqrt{\delta}$

        \item The angle $\alpha = \alpha(u v,p q)$ between the
        line supporting $u v$ and the line supporting $p q$ is
        $\leq \sqrt{\delta}/3$.
    \end{itemize}
    Then, $\diam(Q) \geq \max \brc{ |q u|, |p v|. |p u|, |q
       v|} > 1+ \delta$.

    \lemlab{packing}
\end{lemma}

\begin{proof}
    Let $H_p, H_q$ be the two hyperplanes passing through
    $p,q$, respectively, and perpendicular to the segment
    $p q$.  Let $S$ be the strip bounded by $H_p$ and $H_q$.

    If $u \notin S$, and assume that $u$ is closer to $p$
    than to $q$, then $|q u|$ is minimized when $u$ lies on
    $H_p$. Thus,
    \[
    |q u| = \sqrt{ |p q|^2 + |p u|^2}
    \geq \sqrt{|p q|^2 + 3^2 \delta} = \sqrt{(1-\delta)^2 + 9
       \delta} > \sqrt{(1+\delta)^2} = 1+ \delta,
    \]
    since $|p u| \geq 3\sqrt{\delta}$.

    The case that $v \notin S$ is handled in a similar
    fashion. Thus, we remain with the case $u,v \in S$. Let
    $\ell$ be the line passing through $u$ and parallel to
    $p q$ (we assume here that $u$ is closer to $p q$ than
    $v$). Let $v'$ be the projection of $v$ on $\ell$. It is
    easy to verify that $|p v'| \leq |p v|$, $|q v'| \leq
    |q v|$. See \figref{angle}.

    Furthermore,
    \[
    |u v'| = |u v| \cos \alpha
    \geq \pth{ 1 - \frac{\alpha^2}{2}} |u v|
    \geq \pth{ 1 - \frac{\delta}{18}} |u v|
    \geq \pth{ 1 - \frac{\delta}{18}} (1-\delta)
    \geq 1 - 2 \delta.
    \]
    Assume that $p$ is closer to $u$.
    It is easy verify that in such a case, $|p v'|$ is
    minimized when $u \in H_p$, and $|p u| = 3\sqrt{\delta}$.
    Indeed, $u,v'$ are by construction in $S$. Thus,
    $|p v'|$ is minimized, when  $|p u| = 3 \sqrt{\delta}$. In
    turn, this case is minimized when the angle $\alpha(p u,
    u v') = \pi/2$, which implies that $u \in H_p$. However,
    \[
    |p v'| = \sqrt{ |p u|^2 + |u v'|^2} \geq \sqrt{ 9 \delta +
     (1 -2\delta)} > 1 + \delta,
    \]
    and the result follows.
\end{proof}

We cluster the pairs of $S_i'$ as follows: We pick any pair
$(u_1, v_1) \in S_i'$ and let $U_1$ be the set of all pairs
$(u', v')$ in $S_i'$ such that, $|c(u)c(u')|, |c(v)c(v')|
\leq \rho_i$, where
\[
\rho_i = 100 \sqrt{\diam d r_i} + 54d r_i.
\]
We continue in similar fashion, clustering $S_i' \setminus
U_1$. Let $U_1, \ldots, U_{k_i}$ be the set of clusters of
$S_i'$, where $k_i$ denotes the number of resulting
clusters. Let $(u_j, v_j)$ denote the pair of $S_i'$ that
created $U_j$.

\begin{lemma}
    The number of pairs in a cluster $U_j$ is
    $O(2^{i(d-1)})$, for $j=1, \ldots, k_i$.
    \lemlab{pairs:in:cluster}
\end{lemma}

\begin{proof}
    A cluster $U_j$ have pairs with bounding boxes with
    distance $\leq \rho_i$ from the centers $p_j,q_j$ of
    $\R(u_j)$, $\R(v_j)$, respectively, where $(u_j, v_j)$
    is the pair of nodes defining the cluster $U_j$.

    Let $W \subseteq N_i$ be the set of all the nodes of $T$
    that participate in $U_j$ and that are in distance at
    most $\rho_i$ from $p_j$.

    \begin{figure}
        \centerline{\XIpe{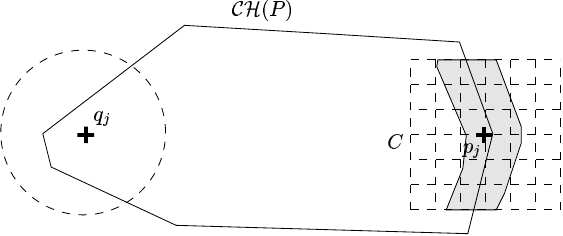}{figs/cluster.ipe}}
        \caption{Illustration of the proof of
           \lemref{pairs:in:cluster}. The cells that
        participate in a cluster $U_j$ are centered around
        $p_j$ and $q_j$. By \lemref{close:to:boundary} all
        such nodes must lie close to the surface of the
        convex-hull $\CH(P)$.}
        \figlab{cluster}
    \end{figure}
    Let $C$ be a cube side length $2\rho_i$, and let $G$ be
    the grid created inside $C$ by breaking it into cubes of
    side length $r_i$. All the nodes of $W$ must intersect
    $C$ (and thus one of the cells of $G$), and are in
    distance at most $21d r_i$ from the boundary of $\CH(P)$,
    by \lemref{close:to:boundary}. See \figref{cluster}. By
    \lemref{grid:surface} the number of grid cells of $G$
    that are in such distance from $\bd{\CH(P)}$ is
    \[
    O\pth{ (\rho_i/r_i)^{d-1}( 1 +  (21d r_i/r_i)^d)
       } =
    O \pth{ 2^{i(d-1)/2}}.
    \]
    But on the other hand, each such cell grid of $G$,
    intersects at most $K(r_i,d,i) = O((r_i2^i/\diam)^d) =
    O(1)$ cells of $N_i$ (and thus of $W$) by
    \lemref{N:i:packing}. In particular, the number of nodes
    in $W$ is $O(2^{i(d-1)/2} K(r_i,d,i)) = O(2^{i(d-1)/2})$.

    Using a similar argument, we can bound the number of
    nodes of $N_i$ that appear in the other side of $U_j$.
    In particular, the overall number of pairs in $U_j$, is
    the product of those two quantities, which is
    $O((2^{i(d-1)/2})^2) = O(2^{i(d-1)})$.
\end{proof}

\begin{lemma}
    For any two clusters $U_j, U_m$, the angle between
    $c(u_j)c(v_j)$ and $c(u_m)c(v_m)$ is larger than
    $\beta$, for $j \neq m$, where $(u_j,v_j), (u_m, v_m)$
    are two pairs defining the clusters $U_j, U_m$, and
    \[
    \beta = \sqrt{\frac{3d r_i}{\diam}}.
    \]
    \lemlab{dirs}
\end{lemma}

\begin{proof}
    Assume that $\alpha$, the angle between $c(u_j)c(v_j)$
    and $c(u_m)c(v_m)$, is smaller than $\beta$. We remind
    the reader, that the clustering radius is $\rho_i = 100
    \sqrt{\diam d r_i} + 54d r_i$.  And let $\delta = 27d
    r_i/\diam$. If all the lengths of all the pairs in $F =
    \brc{ c(u_j), c(v_j) } \times \brc{c(u_m), c(v_m)}$ are
    all larger than $\rho_i/2 \geq 3\sqrt{\delta}\diam$,
    then by \lemref{packing}, we know that the length of one
    of the pairs $(u,v)$ of $F$ is larger than $\diam(1+
    \delta)$.  Which is impossible, because then there are
    two points $p \in \PSET(u), q \in \PSET(v)$ such that
    $|p q| \geq \diam(1+\delta) -4r_i > \diam + 20d r_i >
    \diam$. A contradiction.

    Thus, there must be a pair in $F$ which is shorter than
    $\rho_i/2$, and assume without loss of generality that
    $|c(u_j)c(u_m)| < \rho_i/2$. However, since $(u_m, v_m)$
    is not in the cluster of $(u_j,v_j)$, it follows that
    $|c(v_j)c(v_m)| > \rho_i$.  Using simple geometric
    arguments, it follows that
    \[
    \frac{\alpha}{2} \geq \sin\frac{\alpha}{2} \geq
    \frac{\rho_i -2\delta\diam}{2} \geq 50
    \sqrt{\frac{d r_i}{\diam}} \geq \beta.
    \]
    A contradiction.
\end{proof}

\begin{lemma}
    The number of pairs created by $\AprxDiamWSPD$ is
    $O(1/\eps^{3(d-1)/2})$.

    \lemlab{overall:pairs}
\end{lemma}

\begin{proof}
    By \lemref{dirs}, the angle between the axes of any
    two clusters $U_j, U_m$ is larger than
    $\beta = \sqrt{\frac{3d r_i}{\diam}}$. This implies, by
    using a packing argument on the sphere of directions in
    $\Re^d$, that the number of clusters $k_i$ is at most
    \[
    O \pth{ \frac{1}{\beta^{d-1}}} = O \pth{ \pth{
    \frac{\diam}{r_i}}^{(d-1)/2} } = O \pth{ 2^{i(d-1)/2} },
    \]
    by \obsref{r:i}.

    By \lemref{pairs:in:cluster}, the number of pairs of
    $S_i'$ in each cluster is $O(2^{i(d-1)})$. We conclude,
    that the number of pairs in $S_i$ is
    $O(2^{3i(d-1)/2})$.
    Summing over $i=1, \ldots, \ceil{ \log \frac{ 4
          d}{\eps}}$ (using \lemref{S:i:empty}), we
    conclude, that the overall number of pairs created by
    the algorithm $\AprxDiamWSPD$ is $O(1/\eps^{3(d-1)/2})$.
\end{proof}

\begin{theorem}
    \thmlab{bound}%
    The algorithm $\AprxDiamWSPD$ $\eps$-approximates the
    diameter in
    \[
    O\pth{\pth{n+\frac{1}{\eps^{3(d-1)/2}}}\log{\frac{1}{\eps}}}
    \]
    time.
\end{theorem}

\section{Experimental Results}
\seclab{results}

We had implemented a number of algorithms for computing the
diameter. The algorithms were implemented using \verb|C++|
on a Pentium III 800MhZ running Linux with 512MB memory. Below
we detail the implemented algorithms, and the results
themselves are shown in Tables \ref{tbl:results:mod},
\ref{tbl:results:syn}. Note, that the running time of the
naive exact algorithm is only estimated (this was done by
running the algorithm for $10,000$ points and scaling the
running time according to the number of points at hand).

\subsection{Algorithms Implemented}

The source code of the program is available in
\cite{h-scpca-00}. In particular, an independent module that
implements the algorithm is provided to be used by other
programs.

\subsubsection{Approximation Algorithms that can Compute the
   Exact Solution}

In the following, FS stands for Fair-Split tree and WSPD
stand for Well Separated Pairs Decomposition \cite{ck-dmpsa-95}.
\begin{itemize}
    \item {\bf FS Heap} - implementation of
    \AprxDiameter{}.

    \item {\bf FS Heap Ptr} - implementation of
    \AprxDiameter{} so that the algorithm manipulates
    pointers to the points during the tree construction,
    instead of copying the points themselves. This
    modification results in an improvement in execution
    time, mainly because of a slightly more careful
    implementation of this module.

    \item {\bf FS WSPD} - implementation of \AprxDiamWSPD{}.

    \item {\bf FS Levels} - A variant of {FS WSPD} that
    avoids the usage of a heap by using a FIFO (first in,
    first out) queue (i.e., in the $i$-th iteration all the
    pairs of depth $i$ in the tree are split).

    \item {\bf FS Levels Ptr} - implementation of {FS
       Levels} so that the algorithm manipulates pointers
    to the points during the tree construction, instead of
    copying the points themselves. This modification results
    in a minor improvement in execution time.
\end{itemize}

\subsubsection{\TPDF{$\eps$}{eps}-Approximation Algorithms}

\begin{itemize}
    \item {\bf Chan} - implementation of the algorithm of
    \cite{c-adwse-02}.

    \item {\bf Chan Mod} - implementation of the algorithm
    of \cite{c-adwse-02}, where the 2d recursive call is
    computed using a 2d variant of \AprxDiameter{}. This
    variant seems to be slower than the original algorithm.

    \item {\bf Grid} -  First do grid cleaning, and then use
    {FS Levels} on the resulting point set.

    \item {\bf FS Directions} - A variant of {FS Heap}.
    When a pair $(u,v)$ of nodes is created, so that
    $\nu(u,v) \leq \sqrt{\eps}$, we project the points in
    the two pairs to the axis of the pair, and find the two
    extreme points along the projection, where $\nu(u,v)$ is
    a conservative estimate to the maximum angle between
    $c(u)c(v)$ and any segment $p q$, where $p \in \PSET(u),
    q \in \PSET(v)$.  Return the corresponding pair of
    points, as the candidate of this pair to be the
    diameter.

    \item {\bf Grid FS Dir} - Uses grid cleaning and then
    apply {FS Directions} on the resulting set.
\end{itemize}

\subsubsection{Constant factor approximation to the diameter}

\begin{itemize}
    \item {\bf BBox} - Computes the axis parallel bounding
    box of the input point-set and the two furthest points
    that lie on opposite faces of the bounding box.  This
    algorithm reads the points only once and as such is a
    good reference in comparing the running times of the
    other algorithms.

    \item {\bf PCA} - computes the center of mass, and
    principal component analysis to get three orthogonal
    vectors that represents the input point-set. Scan the
    point-set to find the three extreme pairs in those
    directions, and chooses the longest pair as the
    candidate to be the diameter. Always return a $\sqrt{3}$
    approximation to the diameter.
\end{itemize}

\subsection{Inputs}

We tried two types of inputs: (i) real graphics models -
taken from various publicly-available sources (3D Cafe, and
\cite{lgma-00}). For {\em all} such models the new algorithm
performed extremely well, and the results are shown in Table
\ref{tbl:results:mod}.

We also tried the algorithm on two synthetic inputs for
which it performs badly. The first (denoted as sphere in
Table \ref{tbl:results:syn}) is generated by randomly
sampling points on a sphere. The distribution used was the
uniform distribution (one can prove that the expected
running time of \AprxDiamWSPD{} to compute the exact
diameter for such input in 3d is about $O(n^{3/2})$.).

The second type of input used, were points picked from two
tiny arcs that are far away from each other and are
orthogonal to each other (the supporting tangents of the two
arcs lie are orthogonal). Rotated properly, this input
requires $\Theta(n^2)$ from \AprxDiamWSPD{} to compute the
exact diameter. Even for this case, \AprxDiamWSPD{} is about
twice faster than the naive algorithm (because only $n^2/4$
pairs of points are being tested compared to $n^2/2$ by the
naive algorithm).

\subsection{Analysis of Results}

The new algorithm \AprxDiameter{} (FS Heap) and its variants
perform well in practice. In particular, it computed the
{\em exact} diameter for the models we tested it with, in
time which is at most ten times the time to compute the
axis-parallel bounding box.  Similar performance gains were
seen for $\eps=0.01$.  Chan's algorithm is faster only for
large values of $\eps$, namely $\eps =0.1$.  Even for those
values, the fastest algorithm seems to be {\bf Grid} -
probably because grid cleaning is so effective for such
values of $\eps$. Overall, it seems that the new algorithm
performs especially well for ``real'' inputs.

For the synthetic inputs, the pictures is less encouraging.
Computing the exact diameter by any of the new algorithms
requires a lot of work. For the sphere inputs, once the inputs
are large enough, Chan's algorithms shines being much faster
than any of the alternatives. To some extent, this is the
most favorable input for Chan algorithm: there are no clear
candidates to be the diameter and all feasible directions
are equally feasible.

For the arcs type of inputs, the picture is discouraging
only if one considers the exact variants. Computing the
approximate solution is quite easy, and \AprxDiameter{}
performs extremely well.

\section{Conclusions}
\seclab{conclusions}

We presented a new algorithm for computing the diameter of a
point-set in $2$ and higher dimensions, concentrating on the
3d case for our experiments. The new algorithms can be
easily be implemented in a few hours of work and it
performed extremely well for real inputs, although its
performance on specific synthetic inputs is quadratic.  The
new algorithm can also be used as an approximation
algorithm and then concrete (near linear) guarantees on its
performance are given.  The new algorithm has two favorable
properties: (i) if executed for long-enough time it arrives
sooner or later to the exact solution (this, of course,
holds when executing the algorithm with $\eps=0$), and
furthermore one can give guarantees on the quality of
approximation as a function of the time the algorithm was
executed, and (ii) the algorithm is sensitive to the
hardness of the input and is fast on ``easy'' inputs.

In conclusion, the problem of the diameter seems to be easy
in practice. Even the most naive algorithm (i.e., axis
parallel bounding-box) computes the {\em exact} diameter for
most inputs. Although there were ``real'' inputs for which
only \AprxDiameter{} (and its variants) was able to provide
the exact diameter, in most cases the error was quite small.
Inputs for which the simple algorithms perform really badly
compared with \AprxDiameter{} are rare (i.e., error larger
than $10\%$). However, it is quite easy to generate inputs
for which both the BBox and PCA algorithms give only a
$\sqrt{3}$ approximation.  Thus, the new algorithm should be
used in applications where guarantees on the quality of
approximation are necessary, and speed is a major
consideration.  Otherwise, if the application can manage
with a rough approximation of the diameter it seems that the
axis-parallel bounding box approach might be sufficient.

As mentioned before, the source code of the program used is
the experiments is available from \cite{h-scpca-00}.

The author conjectures that the currently (theoretically)
fastest approximation algorithm for $d>3$, due to Chan
\cite{c-adwse-02}, which runs in $O(n +1/\eps^{d-1})$ time,
can be significantly improved. We leave that as an open
question for further research.

\subsection*{Acknowledgments}

The author wishes to thank Pankaj Agarwal, Boris Aronov,
Jeff Erickson, Piotr Indyk, Lutz Kettner and Edgar Ramos for
helpful discussions concerning the problem studied in this
paper and related problems. Finally, the author thank the
anonymous referees for a number of useful comments.

\BibTexMode{%
   \bibliographystyle{alpha}
   \bibliography{diameter}
}%
\BibLatexMode{\printbibliography}

\begin{table*}[p]
\centerline{
\begin{tabular}{|l||r|r|r|r|r|r|r|r|}
    \cline{2-9}
\multicolumn{1}{c||}{}
& \multicolumn{8}{c|}{Running time in seconds}
\\
\hline
  input
& Clock
& Bunny
& Dragon II
& DS9
& Hand
& Dragon
& {\footnotesize Buddha}
& Blade
  \\
  points
&{\footnotesize 30,744}
&{\footnotesize 35,947}
&{\footnotesize 54,831}
&{\footnotesize 298,517}
&{\footnotesize 327,323}
&{\footnotesize 437,645}
&{\footnotesize 543,652}
&{\footnotesize 882,954}
  \\
  \hline%
  \hline%
  \multicolumn{1}{|l||}{$\eps=0\Bigl.$}\\
  \hline%
FS Heap           & 0.07       & 0.08       & 0.08       & 0.76       & 0.49       & 0.92       & 0.89       & 1.29      \\
FS Heap Ptr & 0.05       & 0.06       & 0.05       & 0.57       & 0.35       & 0.69       & 0.66       & 0.91      \\
FS WSPD           & 0.06       & 0.08       & 0.09       & 0.76       & 0.51       & 0.85       & 0.92       & 1.29      \\
FS Ptr Levels     & 0.06       & 0.08       & 0.09       & 0.63       & 0.41       & 0.83       & 0.79       & 1.18      \\
FS Levels         & 0.07       & 0.09       & 0.11       & 0.80       & 0.52       & 0.94       & 0.99       & 1.52      \\
  \hline%
  \hline%
  \multicolumn{1}{|l||}{$\eps=0.01\Bigl.$}\\
  \hline%
FS Heap & 0.07       & 0.06       & 0.07       & 0.74       & 0.50       & 0.82       & 0.88       & 1.28      \\
FS Heap Ptr & 0.05       & 0.05       & 0.05       & 0.54       & 0.35       & 0.64       & 0.63       & 0.91      \\
FS WSPD        & 0.06       & 0.06       & 0.07       & 0.74       & 0.48       & 0.79       & 0.87       & 1.27      \\
FS Ptr Levels  & 0.06       & 0.07       & 0.08       & 0.61       & 0.41       & 0.80       & 0.77       & 1.16      \\
FS Levels       & 0.07       & 0.09       & 0.11       & 0.74       & 0.52       & 0.89       & 0.95       & 1.49      \\
Chan           & 0.95       & 1.42       & 1.99       & 3.67       & 6.75       & 9.40       & 11.30      & 12.18     \\
Chan Mod       & 0.87       & 1.25       & 1.83       & 3.55       & 6.60       & 9.27       & 11.23      & 12.05     \\
Grid           & 0.10       & 0.15       & 0.18       & 0.48       & 0.58       & 0.97       & 1.02       & 1.22      \\
FS Directions  & 0.32       & 0.18       & 0.11       & 13.61      & 0.56       & 1.61       & 1.46       & 1.43      \\
  \hline%
  \hline%
  \multicolumn{1}{|l||}{$\eps=0.1\Bigl.$}\\
  \hline%
FS Heap            & 0.05       & 0.05       & 0.07       & 0.61       & 0.46       & 0.69       & 0.74       & 1.12      \\
FS Heap Ptr & 0.02       & 0.04       & 0.05       & 0.43       & 0.31       & 0.53       & 0.55       & 0.77      \\
FS WSPD         & 0.04       & 0.05       & 0.07       & 0.62       & 0.41       & 0.68       & 0.74       & 1.11      \\
FS Ptr Levels   & 0.04       & 0.06       & 0.07       & 0.50       & 0.36       & 0.61       & 0.62       & 0.90      \\
FS Levels        & 0.05       & 0.07       & 0.10       & 0.62       & 0.47       & 0.71       & 0.78       & 1.17      \\
Chan            & 0.06       & 0.09       & 0.08       & 0.26       & 0.28       & 0.40       & 0.50       & 0.75      \\
Chan Mod        & 0.08       & 0.09       & 0.09       & 0.27       & 0.30       & 0.41       & 0.51       & 0.74      \\
Grid            & 0.03       & 0.04       & 0.05       & 0.25       & 0.26       & 0.36       & 0.43       & 0.68      \\
Grid FS Dir     & 0.04       & 0.03       & 0.05       & 0.25       & 0.25       & 0.35       & 0.42       & 0.69      \\
FS Directions   & 0.09       & 0.06       & 0.08       & 1.58       & 0.45       & 0.75       & 0.82       & 1.11      \\
  \hline%
  \multicolumn{2}{|c|}{$\eps=\frac{\sqrt{3}-1}{\sqrt{3}} \approx 0.73\Bigr.$}\\
  \hline%
BBox           & 0.00       & 0.01       & 0.00       & 0.03       & 0.03       & 0.05       & 0.07       & 0.10      \\
PCA            & 0.01       & 0.01       & 0.02       & 0.08       & 0.08       & 0.12       & 0.14       & 0.22      \\

\hline
\end{tabular}}
\caption{Results for various models}
\label{tbl:results:mod}
\end{table*}

\begin{table*}[p]
    \centerline{
\begin{tabular}{|l||r|r|r|r||r|r|r|r|}
    \cline{2-9}
\multicolumn{1}{c||}{}
& \multicolumn{8}{c|}{Running time in seconds}
\\
\cline{2-9}
  \multicolumn{1}{c||}{}
& \multicolumn{4}{c||}{Sphere}
& \multicolumn{4}{c|}{ARCS}
\\
\hline
  input
&  $10^3$
&  $10^4$
&  $10^5$
&  $10^6$
&  $10^3$
&  $10^4$
&  $10^5$
&  $2*10^5$\\
points                & 1,000      & 10,000     & 100,000    & 200,000    & 1,000      & 10,000     & 100,000    & 200,000   \\
  \hline%
  \multicolumn{1}{|l||}{$\eps=0$}\\
  \hline%
  FS Heap           & 0.04       & 1.14       & 50.15      & 202.26     & 0.03       & 0.01       & 0.04       & 0.08      \\
FS Heap Ptr & 0.03       & 2.47       & 144.63     & 450.29     & 0.04       & 0.00       & 0.03       & 0.07      \\
FS WSPD           & 0.04       & 1.19       & 53.74      & 157.08     & 0.03       & 0.00       & 0.05       & 0.09      \\
FS Ptr Levels     & 0.03       & 1.00       & 49.55      & 179.26     & 0.03       & 0.00       & 0.04       & 0.08      \\
FS Levels         & 0.03       & 1.05       & 52.02      & 186.01     & 0.03       & 0.01       & 0.06       & 0.11      \\
  \hline%
  \multicolumn{1}{|l||}{$\eps=0.01$}\\
  \hline%
FS Heap & 0.07       & 0.96       & 34.32      & 140.69     & 0.06       & 0.00       & 0.03       & 0.09      \\
FS Heap Ptr & 0.06       & 2.11       & 126.57     & 371.62     & 0.07       & 0.00       & 0.03       & 0.08      \\
FS WSPD        & 0.03       & 1.00       & 37.12      & 96.99      & 0.00       & 0.01       & 0.05       & 0.09      \\
FS Ptr Levels  & 0.03       & 0.82       & 33.85      & 122.99     & 0.00       & 0.00       & 0.05       & 0.09      \\
FS Levels       & 0.04       & 0.88       & 35.62      & 130.08     & 0.00       & 0.01       & 0.06       & 0.10      \\
Chan           & 0.04       & 0.50       & 3.48       & 6.45       & 0.00       & 0.06       & 0.13       & 0.20      \\
Chan Mod       & 0.03       & 0.81       & 3.82       & 6.96       & 0.00       & 0.07       & 0.13       & 0.20      \\
Grid           & 0.21       & 0.99       & 39.68      & 125.04     & 0.06       & 0.03       & 0.10       & 0.16      \\
FS Directions  & 0.33       & 1.41       & 440.65     & 961.95     & 0.06       & 0.01       & 0.05       & 0.10      \\
  \hline%
  \multicolumn{1}{|l||}{$\eps=0.1$}\\
  \hline%
FS Heap            & 0.05       & 0.15       & 0.35       & 0.70       & 0.02       & 0.00       & 0.05       & 0.10      \\
FS Heap Ptr & 0.03       & 0.42       & 5.17       & 6.21       & 0.00       & 0.00       & 0.03       & 0.07      \\
FS WSPD         & 0.02       & 0.18       & 0.33       & 0.64       & 0.00       & 0.00       & 0.05       & 0.09      \\
FS Ptr Levels   & 0.02       & 0.12       & 0.35       & 0.74       & 0.00       & 0.01       & 0.04       & 0.08      \\
FS Levels        & 0.02       & 0.13       & 0.32       & 0.67       & 0.00       & 0.00       & 0.05       & 0.11      \\
Chan            & 0.02       & 0.07       & 0.17       & 0.26       & 0.00       & 0.01       & 0.07       & 0.16      \\
Chan Mod        & 0.02       & 0.06       & 0.17       & 0.26       & 0.00       & 0.01       & 0.07       & 0.15      \\
Grid            & 0.01       & 0.16       & 0.26       & 0.34       & 0.00       & 0.01       & 0.08       & 0.15      \\
Grid FS Dir     & 0.01       & 0.29       & 0.40       & 0.49       & 0.00       & 0.00       & 0.07       & 0.15      \\
FS Directions   & 0.03       & 0.46       & 8.16       & 16.72      & 0.00       & 0.01       & 0.04       & 0.09      \\
  \hline%
  \multicolumn{1}{|l||}{$\eps=\frac{\sqrt{3}-1}{\sqrt{3}} \approx 0.73$}\\
  \hline%
BBox           & 0.02       & 0.00       & 0.01       & 0.02       & 0.00       & 0.00       & 0.02       & 0.02      \\
PCA            & 0.02       & 0.00       & 0.03       & 0.05       & 0.00       & 0.01       & 0.02       & 0.05      \\

\hline
\end{tabular}
}
\caption{Results for syntetic inputs}
\label{tbl:results:syn}
\end{table*}

\end{document}